\documentclass[12pt,draftclsnofoot,journal,letterpaper,onecolumn]{IEEEtran}

\usepackage{enumitem}
\usepackage{lipsum}

\setlist[enumerate]{wide=\parindent}

\usepackage[cmex10]{amsmath}
\usepackage{amsmath, optidef}
\usepackage{amssymb}
\usepackage{amsthm}

\usepackage{lipsum}
\usepackage[T1]{fontenc}
\usepackage{mathtools}   
\usepackage{cases}
\usepackage{stfloats}
\usepackage{amsmath,amssymb}
\DeclareMathOperator{\sinc}{sinc}
\usepackage{mathtools} 
 \usepackage{mathtools} 

\usepackage{bm}
\usepackage{amsbsy}
\usepackage{epsfig}

\usepackage{graphicx}
\usepackage[outdir=./]{epstopdf}
\usepackage{subfigure}
\usepackage{xcolor}
\usepackage[colorinlistoftodos]{todonotes}
\usepackage{pdfpages}
\usepackage{amsmath}
\usepackage{bm}
\usepackage{caption}
\usepackage[labelsep=period]{caption}

\newtheorem{lemma}{Lemma}

\usepackage[english]{babel}
\usepackage{algorithm}
\usepackage{algpseudocode}
\usepackage{optidef}
\usepackage{siunitx}
\begin{document}
	
	\title{Parametric Channel Estimation for LoS Dominated Holographic Massive MIMO Systems}
	\author{Mojtaba Ghermezcheshmeh, Vahid Jamali,\textit{ Member, IEEE}, Haris Gacanin,\textit{ Fellow, IEEE},  and Nikola Zlatanov,\textit{ Member, IEEE}
    \footnote{M. Ghermezcheshmeh, and N. Zlatanov are with the Department of Electrical and Computer Systems Engineering, Monash University, Melbourne, VIC 3800, Australia (Emails: mojtaba.ghermezcheshmeh@monash.edu; nikola.zlatanov@monash.edu).
    
    V. Jamali is with the Department of Electrical Engineering and Information Technology, Technical University of Darmstadt, Darmstadt, 64283, Germany (Email: vahid.jamali@tu-darmstadt.de).
    
    H. Gacanin is with the Faculty of Electrical Engineering and Information Technology, RWTH Aachen University, Aachen, Germany (Email: harisg@ice.rwth-aachen.de).}
	
	}

		\maketitle
\begin{abstract}
Holographic massive multiple-input multiple-output (MIMO), in which a spatially continuous surface is being used for signal transmission and reception, have emerged as a promising solution for improving the coverage and data rate of wireless communication systems.
 To realize these objectives, the acquisition of accurate channel state information in holographic massive MIMO systems is crucial. In this paper, we propose a channel estimation scheme based on a parametric physical channel model for line-of-sight (LoS) dominated communication in millimeter and terahertz wave bands. The proposed channel estimation scheme exploits the specific structure of the radiated beams generated by the continuous surface to estimate the channel parameters in a dominated LoS channel model. The training overhead and computational complexity of the proposed scheme do not scale with the number of antennas. 
The simulation results demonstrate that the proposed estimation scheme significantly outperforms other benchmark schemes in poor scattering environment.


\end{abstract}


\section{Introduction}
The millimeter-wave (mmWave) and terahertz (THz) frequency bands will play a pivotal role in next-generation (e.g., beyond 5G) wireless networks since they would be able to provide abundant spectrum resources, higher data rates,
 and lower latency \cite{chaccour2021seven}, \cite{tang2021path}. However, severe path-loss, atmospheric absorption, human blockage, and other environmental obstruction are key challenges for enabling wireless communication in the mmWave and THz bands \cite{wan2021terahertz}, \cite{jamali2020intelligent}. 
 Over the past few years, massive multiple-input multiple-output (MIMO) communication systems, in which base stations (BSs) are equipped with a large antenna array, have been introduced to provide beamforming gains that can overcome the mentioned limitations \cite{wei2014key}, \cite{larsson2014massive}.   
 Since the beamforming gain of massive MIMO increases with the number of antennas, it would be highly desirable to have as many antennas as possible that are compactly arranged 
 \cite{zhu2020stochastic}, \cite{bjornson2019massive}. 
In this direction, spatially continuous apertures with densely deployed antennas, known as holographic massive MIMO \cite{bjornson2019massive} and large intelligent surfaces (LISs) \cite{dardari2020communicating}, were introduced to go beyond massive MIMO systems \cite{huang2020holographic}-\cite{hu2018beyond}. Since these continuous apertures can be seen as an extension of traditional massive MIMO from a discrete array to a continuous surface \cite{hu2018beyond}, they have been referred to as holographic massive MIMO \cite{bjornson2019massive}, \cite{wan2021terahertz}.

The holographic massive MIMO that comprises a massive number of phase-shifting elements with steerable beamforming weight can be used for generating narrow beams with high beamforming gains \cite{shlezinger2021dynamic}, see Fig. \ref{Fig1}. Moreover, in the holographic massive MIMO, multiple radio frequency (RF) chains can be connected to the aperture to enable spatial multiplexing \cite{bjornson2019massive}, \cite{shlezinger2021dynamic}. The existing research works have verified the effect of the holographic massive MIMO systems on enhancing the communication performance in various scenarios under the assumption of the availability of perfect channel state information (CSI)  \cite{jung2021performance}-\cite{hu2018beyond1}. To achieve the full potential of the holographic massive MIMO systems, the acquisition of accurate CSI is a fundamental and challenging task in practice.  Since the holographic massive MIMO consists of a massive number of phase-shifting elements, estimating all the entries of the channel matrix is not feasible in practice since it results in substantial training overhead and computational complexity.




\subsection{Related Work}
 Up to now, a variety of channel estimation schemes have been proposed for the massive MIMO communication
systems, including exhaustive search \cite{dai2006efficient}, hierarchical search \cite{xiao2016hierarchical}-\cite{chen2018beam}, and compressed sensing (CS) \cite{lee2016channel}-\cite{sanchez2021gridless}.
The authors in \cite{dai2006efficient} propose the exhaustive search algorithm as a straightforward approach where the transmitter and receiver scan all possible angular directions to find the best pair of AoA and AoD.
 However, the training overhead of the exhaustive search approach is prohibitively high, especially when a large number of antennas generate narrow beams in the mmWave and THz frequency bands. 
 To improve the efficiency of the exhaustive search, the hierarchical search based on a predefined codebook was proposed in \cite{xiao2016hierarchical}-\cite{chen2018beam}. In the first stage of the hierarchical search, the codewords with larger beam widths are used to scan the entire angular domain. Then, in the second stage, codewords with narrower beam widths are used to scan only a specific range obtained at the first stage. 
 In \cite{xiao2016hierarchical}, a hierarchical codebook
is designed, where sub-array
and deactivation antenna processing techniques were exploited to generate the codebook via closed-form expressions. However, turning off some antennas may not be a good approach since the reduced array gain has an undesirable effect on the performance \cite{zhang2017codebook}.
The authors in \cite{noh2017multi} design  multi-resolution beamforming sequences to quickly search out the dominant channel direction. Reference \cite{chen2018beam} proposes a beam training method based on dynamic hierarchical codebook to estimate the mmWave massive MIMO channel
with multi-path components.
However, the accuracy of the hierarchical search is limited by the codebook size.
Moreover, all these hierarchical schemes  may incur high training overhead and system latency because they require non-trivial coordination among the transmitter and the receiver.

At mmWave frequency, typically only a few dominant path components contribute to the received power, i.e., mm-wave channels are sparse in the angular domain. Thereby, different compressed sensing (CS) methods exploit the sparsity of the mmWave channels to estimate the channel with relatively few measurements despite a large number of phase-shifting elements \cite{lee2016channel}-\cite{sanchez2021gridless}.
In the CS-based channel estimation, first a set of random training sequences is used to measure the channel, and then a sparse recovery algorithm is employed to obtain the path parameters. The CS-based channel estimation schemes can be classified into two main categories: on-grid and gridless estimation.
The on-grid schemes assume that the angle of arrivals (AoAs) and angle of departures (AoDs) are discrete with finite resolutions although the actual angles are continuous in practice \cite{lee2016channel}-\cite{lim2020efficient}. 
The accuracy of these schemes depends on the quantization resolution as well as the number of measurements \cite{gao2015spatially}. As a result, the training overhead can be significantly high in order to achieve satisfactory accuracy. On the other hand, the gridless schemes deal directly with the continuous domain without imposing a discrete dictionary \cite{pejoski2015estimation}-\cite{sanchez2021gridless}. The gridless schemes increase the accuracy of estimation at the expense of increased computational complexity. 
None of these channel estimation schemes have been explicitly designed for channels dominated by a line-of-sight (LoS) path. The question that we would like to answer in this paper is what if the channel is dominated by a LoS path. Answering this question is important because in future mmWave and THz communications the channel is sparse and LoS dominant.

In the user localization problem, the parametric channel model is used to estimate the distance of the user from the BS and the elevation and azimuth AoD \cite{shafin2017angle}, \cite{he2020channel2}. 
Thereby, the channel estimation problem and user localization are highly correlated \cite{li2021millimeter}, \cite{wang2017unified}.  Different from the studies on localization problems \cite{wymeersch2020radio}-\cite{guidi2021radio}, in this work, we circumvent the estimation of distance in order to obtain a closed-form solution for the channel estimation problem for channels dominated by a LoS path. 
 
 



\subsection{Main Contributions}
In this paper, we focus on channels that are comprised of a LoS path, and aim at developing a channel estimation strategy that fully exploit the structure of the radiated pattern to significantly reduce the channel estimation overhead.
Different from the aforementioned works, we obtain a closed-form solution for the channel estimation problem based on a parametric channel model. 
The training overhead and the computational complexity of the proposed scheme do not scale with the the number of antennas. Specifically, our main contributions are as follows.

\begin{itemize}
  \item 
Based on the parametric channel model, we show that only three parameters are required to be estimated for channels dominated by a LoS path. These three parameters include the distance from the BS to the user and the elevation and azimuth AoD. However, the distance from the BS to the user has to be estimated very precisely, which is not possible in practice. Therefore, we simplify the channel for the far-field region of the BS to circumvent the estimation of the distance parameter. We show that if we use the simplified model, the number of parameters needed to be estimated decreases from three to two.
 
  

  \item
 For the near-field of the BS, we partition the continuous aperture into tiles such that the far-field condition holds for each tile \cite{najafi2020physics}. Then, the channel in the near-field is modeled as the superposition of the channels through the individual tiles. We show that the channel models in the far-field and near-field regions have similar structures, which can be written in terms of sinc functions for linear phase shift profiles at the antennas. 

 \item Finally, we exploit the specific structure of the radiated beams generated by the continuous aperture to propose the channel estimation scheme. 
 Since the performance of the proposed iterative algorithm depends on the initial values of the algorithm, we propose a simple method to provide initial values using three pilot signals. We numerically demonstrate that the proposed estimation scheme significantly outperforms other benchmark schemes.


\end{itemize}


\begin{figure}[t]
\centering
\includegraphics[width=0.55\linewidth]{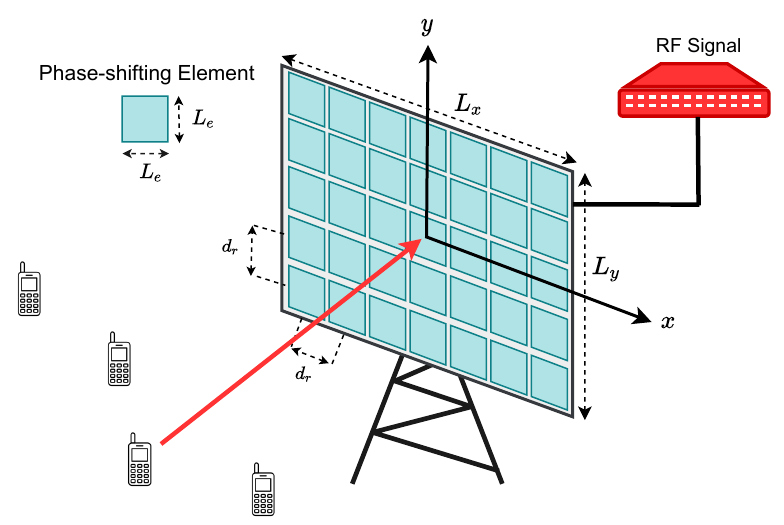}
\caption{\small{Schematic illustration of the holographic MIMO transceiver. Each phase-shifting element of the aperture changes the phase of the RF signal, and hence the BS is able to send a beamformed signal towards the users.}}
\label{Fig1}
\end{figure}

\section{System Model}
\subsection{System Model}
\label{SystemModel}
We consider a holographic massive MIMO system where a BS equipped with a rectangular aperture is serving multiple users with a single antenna, see Fig. \ref{Fig1}. For data communication, multiple RF chains can be used at the BS to transmit and receive a superposition of beams to serve multiple users using spatial multiplexing. However, we use only one of these RF chains at the BS for the channel estimation. We assume that the users transmit orthogonal pilot symbols than are then used for the channel estimation at the BS.


As shown in Fig. \ref{Fig1}, the size of the rectangular aperture is $L_x \times L_y$, where $L_x$ and $L_y$ are the width and the length of the aperture, respectively. The aperture is comprised of a large number of sub-wavelength phase-shifting elements of size $L_e \times L_e$ that can change the phase of the signal.
We assume that the aperture lies in the $x-y$ plane of a Cartesian coordinate system, and its center is placed at the origin of the coordinate system, see Fig. \ref{Fig1}. Let $d_r$ be the distance between neighboring phase-shifting elements. 
The total number of phase-shifting elements of the aperture is given by $M=M_x \times M_y$, where $M_x=L_x/d_r$ and $M_y=L_y/d_r$. Assuming $M_x$ and $M_y$ are odd numbers, the position of the $(m_x,m_y)$-th phase-shifting element is given by $(x,y)=(m_xd_r,m_yd_r)$ for $m_x=-\frac{M_x-1}{2},\cdots, \frac{M_x-1}{2}$ and $m_y=-\frac{M_y-1}{2},\cdots, \frac{M_y-1}{2}$.

In the uplink communication direction, the baseband received signal at the BS, denoted by $y$, can be expressed as follows:
 \begin{align}
    y = \text{vec}\left(\mathbf{\Gamma}\right)^\text{T}\left( \text{vec}\left(\mathbf{H}\right) x + \mathbf{n}\right),
     \label{signal_model}
 \end{align} 
where $\mathbf{\Gamma} \in \mathbb{C}^{M_x \times M_y}$ is the beamforming weight matrix at the BS, $\mathbf{H}\in \mathbb{C}^{M_x \times M_y}$ is the channel matrix from the user to the phase-shifting elements, $x$ is the transmitted signal from the user, and $\mathbf{n}$ is the Gaussian noise vector. The $(m_x+\frac{M_x+1}{2})$-th row and $(m_y+\frac{M_y+1}{2})$-th column of $\mathbf{\Gamma}$, denoted by $ [\mathbf{\Gamma}]_{\left(m_x+\frac{M_x+1}{2},m_y+\frac{M_y+1}{2}\right)} $, represents the beamforming weight at the $(m_x,m_y)$-th phase-shifting element. Similarly, the $(m_x+\frac{M_x+1}{2})$-th row and $(m_y+\frac{M_y+1}{2})$-th column of $\mathbf{H}$, denoted by $ [\mathbf{H}]_{\left(m_x+\frac{M_x+1}{2},m_y+\frac{M_y+1}{2}\right)} $, represents the channel coefficient between the $(m_x,m_y)$-th phase-shifting element and the user.


\subsection{Channel Model}
\label{ChannelModel}
In this subsection, we model the channel between the BS and the users based on the path parameters of the system. Since the users send orthogonal pilot symbols for channel estimation, we model the channel between the BS and a typical user. In addition to the LoS channel between the BS and the user, there are $Q$ scatters, where each scatter is assumed to contribute a single propagation path between the BS and the user. Therefore, the BS-user channel consists of one LoS and $Q$ non-LoS (NLoS) components.

Let $d_{m_xm_y}$ and $l_{m_xm_y,q}$ denote the distance from the user to the center of the$(m_x,m_y)$-th phase-shifting element at the BS through the LoS path and $q$-th scatter, respectively. 
 The channel coefficient between the $(m_x,m_y)$-th phase-shifting element and the user, $  [\mathbf{H}]_{\left(m_x+\frac{M_x+1}{2},m_y+\frac{M_y+1}{2}\right)}  $, is represented as 
 \begin{align}
   [\mathbf{H}]_{\left(m_x+\frac{M_x+1}{2},m_y+\frac{M_y+1}{2}\right)}  =\underbrace{ \text{PL}\left(d_{m_xm_y}\right)e^{-jk_0d_{m_xm_y}}}_\text{LoS component} + \underbrace{\sum_{q=1}^{Q} \zeta_q \text{PL}\left(l_{m_xm_y,q}\right) e^{-jk_0l_{m_xm_y,q}}}_\text{NLoS components}, 
     \label{ITSelement_user}
 \end{align} 
 where $k_0=2\pi/\lambda$ is the wave number, $\lambda$ is the wavelength of the carrier frequency, $\zeta_q$ is the small-scale fading of the BS-user channel through the $q$-th scatter, and $\text{PL}\left(.\right)$ is the channel path loss function. According to \cite{ellingson2019path}, the channel path loss for the LoS path can be expressed as 
 \begin{align}
\text{PL}\left(d_{m_xm_y}\right)=\frac{\lambda \sqrt{F_{m_xm_y}}}{4\pi d_{m_xm_y}},
     \label{LoS_component}
 \end{align} 
 where $F_{m_xm_y}$ accounts for the effect of the size and power radiation pattern of the phase-shifting elements on the path gain. 

Due to severe path-loss at the mmWave and THz frequency bands, the power of the LoS component is much higher than the power of  the NLoS component \cite{akdeniz2014millimeter}, \cite{muhi2010modelling}. In fact, it is less likely to build an effective communication via NLoS components in presence of a strong LoS link \cite{wang2021joint}. Thereby, we focus on the LoS component to obtain a closed form expression for the BS-user channel\footnote{The effect of NLoS components on the proposed scheme will be considered as an interference term and investigated in the Numerical Results section.}. As a result, from (\ref{ITSelement_user}) and (\ref{LoS_component}), the LoS component of the channel coefficient between the $(m_x,m_y)$-th phase-shifting element and the user can be written as
 \begin{align}
   [\mathbf{H}]_{\left(m_x+\frac{M_x+1}{2},m_y+\frac{M_y+1}{2}\right)} =\frac{\lambda \sqrt{F_{m_xm_y}}}{4\pi d_{m_xm_y}}\left(d_{m_xm_y}\right)e^{-jk_0d_{m_xm_y}}.
     \label{hLoS component}
 \end{align}


 
 Each phase-shifting element at the aperture can be configured to impose different levels of phase shifts on the transmitted and received signal \cite{shlezinger2021dynamic}. Let $ [\mathbf{\Gamma}]_{\left(m_x+\frac{M_x+1}{2},m_y+\frac{M_y+1}{2}\right)}=  e^{j\beta_{m_xm_y}}$ denote the beamforming weight of the $(m_x,m_y)$-th phase-shifting element at the aperture, where $\beta_{m_xm_y}$ is the phase shift at the $(m_x,m_y)$-th element. 
For ease of presentation, let us define the phase shift parameters of the aperture for all elements as 
$\bm{\beta}= \left\{\beta_{m_xm_y};\, \forall m_x, m_y \right\}$.
Finally, considering the effect of all phase-shifting elements, the effective BS-user LoS channel, denoted by $G(\bm{\beta})\triangleq \text{vec}\left(\mathbf{\Gamma}\right)^\text{T} \text{vec}\left(\mathbf{H}\right)$, can be written as
\begin{align}
   G\left(\bm{\beta}\right) & = \sum_{m_x=-\frac{M_x-1}{2}}^{\frac{M_x-1}{2}}\sum_{m_y=-\frac{M_y-1}{2}}^{\frac{M_y-1}{2}}  [\mathbf{\Gamma}]_{\left(m_x+\frac{M_x+1}{2},m_y+\frac{M_y+1}{2}\right)}  [\mathbf{H}]_{\left(m_x+\frac{M_x+1}{2},m_y+\frac{M_y+1}{2}\right)} \nonumber\\
    & = \frac{ \lambda}{4\pi} \sum_{m_x=-\frac{M_x-1}{2}}^{\frac{M_x-1}{2}}\sum_{m_y=-\frac{M_y-1}{2}}^{\frac{M_y-1}{2}} \frac{\sqrt{F_{m_xm_y}}}{d_{m_xm_y}}e^{-j\left(k_0 d_{m_xm_y}-\beta_{m_xm_y}\right)} .
    \label{endtoend_exact_ITS}
\end{align}
From (\ref{endtoend_exact_ITS}), it is clear that the effective BS-user channel depends on the imposed phase shift by each element at the aperture and the distance of each element from the user.

\subsection{Optimal Phase Shifts at the Aperture}
\label{2C}
\begin{figure}
\centering
    \includegraphics[scale=0.65]{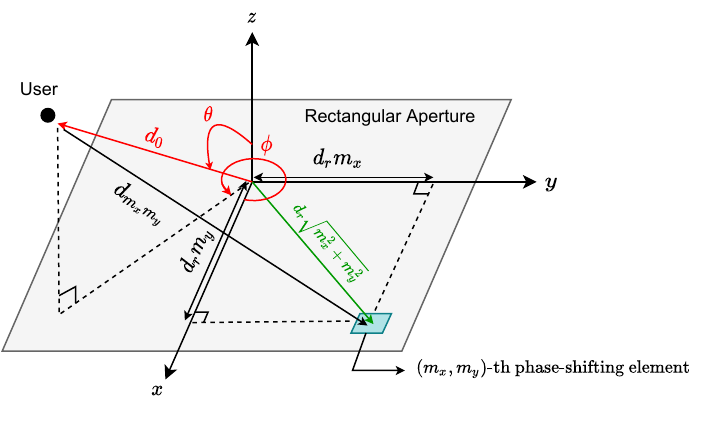}
\caption{\small{Using the cosin rule to obtain the distance between the $(m_x,m_y)$-th phase-shifting element at the aperture to the user. 
}}
\label{fig3}
\end{figure}

According to the effective BS-user channel in (\ref{endtoend_exact_ITS}), it is clear that the effective channel gain of LoS path,  $G\left(\bm{\beta}\right)$, is maximized by setting $\beta_{m_xm_y}$ to 
\begin{align}
    \beta_{m_xm_y}= \text{mod}\left(k_0d_{m_xm_y},2\pi\right), \:\forall \, m_x,m_y.
    \label{nearphase_ITS}
\end{align}
The total number of parameters to be estimated is equal to the total number of phase-shifting elements at the aperture $M_x\times M_y$. 

Let $d_0$ denote the distance between the user and the center of the aperture
.  
In addition, let $\theta \in [0,\pi/2]$ and $\phi \in [0,2\pi]$ denote the elevation and azimuth angles of the impinging wave from the user to the center of the aperture, see Fig. \ref{fig3}. 
Using the cosin rule, the distance between the user and the $(m_x,m_y)$-th phase-shifting element can be written as \cite{antennatheory}
\begin{align}
    d_{m_xm_y} = \Big(d_0^2 + d_r^2\big(m_x^2+m_y^2\big)-2d_0d_r\big(m_x\alpha_{1} 
    +m_y \alpha_{2}
    \big)\Big)^{\frac{1}{2}},
    \label{neard_mxmy1}
\end{align}
where
\begin{align}
   & \alpha_{1}=\sin{\left(\theta\right)}\cos{\left(\phi\right)},
   \label{alpha_1}\\
    &\alpha_{2}=\sin{\left(\theta\right)}\sin{\left(\phi\right)}.
    \label{alpha_2}
\end{align}
Let us emphasize that for any $m_x$ and $m_y$, all the three parameters, $d_0$, $\alpha_1$, and $\alpha_2$ are the same. Substituting (\ref{neard_mxmy1}) into (\ref{nearphase_ITS}), we obtain the final expression for the optimal phase shift at the $(m_x,m_y)$-th element of the aperture that maximizes the effective BS-user channel as
\begin{align}
    \beta_{m_xm_y}= \text{mod}\bigg(k_0\Big(d_0^2 + d_r^2\big(m_x^2+m_y^2\big)-2d_0d_r\big(m_x\alpha_{1} 
    +m_y \alpha_{2}
    \big)\Big)^\frac{1}{2}
    ,2\pi\bigg), \:\forall \, m_x,m_y.
    \label{nearphase_1}
\end{align}
From (\ref{nearphase_1}), we conclude that in order to obtain the optimal phase shifts at the aperture, we need to estimate only three parameters, $d_0, \alpha_{1},$ and $\alpha_{2}$.   Therefore, using (\ref{nearphase_1}) reduces the number of parameters from $M_x \times M_y$ to three. This is because the initial $M_x \times M_y$ parameters are mutually dependent with three independent degrees of freedom $d_0$, $\alpha_1$ and $\alpha_2$. Since these three parameters specify the user location, the channel estimation problem in LoS dominated channel is equivalent with the localization problem. However, $d_0$ has to be estimated very precisely to set the optimal phase shifts, which is not possible in practice. In the following, we circumvent the estimation of $d_0$ such that $\alpha_1$ and $\alpha_2$ are the only two parameters left to estimate.

\section{Far-Field and Near-Field Channel Models}
\label{section4}
In this section, we provide the effective BS-user channel models for the far-field and the near-field regions. Then, in Sec. V, we will utilize the specific structure of the channel model in dominated LoS link to develop our channel estimation scheme.  

\subsection{Far-Field Channel Model}
\label{far_field_approx}
For the sake of completeness, we follow the same approach as [48, Ch. 6] to simplify the BS-user channel in (\ref{endtoend_exact_ITS}) such that it holds for the far-field region of the BS. 
In the antenna array terminology, the far-field is referred to as the condition that the maximum phase error of the received signal on the antenna array does not exceed $\frac{\pi}{8}$ \cite{selvan2017fraunhofer}.

The distance between the user and the $(m_x,m_y)$-th phase-shifting element at the aperture in (\ref{neard_mxmy1}) can be rewritten as 
\begin{align}
d_{m_xm_y}=d_0\Bigg(1 + \Big( \frac{d_r^2\left(m_x^2+m_y^2\right)}{d_0^2} - \frac{2d_r\left(m_x \alpha_{1}+m_y\alpha_{2}\right)}{d_0}\Big)\Bigg)^{\frac{1}{2}}.
    \label{rewrite_d}
\end{align}
Now, we can use the first-order Taylor approximation \cite{jeffrey2007table}, i.e.,  $\sqrt{1+x}\approx1+\frac{x}{2}$, to approximate $d_{m_xm_y}$ in (\ref{rewrite_d}) as
\begin{align}
d_{m_xm_y}\approx d_0 -d_r\left(m_x \alpha_{1}+m_y\alpha_{2}\right)+\frac{d_r^2\left(m_x^2+m_y^2\right)}{2d_0}.
    \label{fard_mxmy}
\end{align}
Substituting (\ref{fard_mxmy}) into (\ref{endtoend_exact_ITS}), the effective BS-user channel in the far-field region of the BS, denoted by $G^{(ff)}(\bm{\beta})$, is then approximated as\footnote{Based on the Lagrange error bound of the Taylor approximation as the worse case \cite{jeffrey2007table}, the maximum error is a function of $\left(\frac{d_r}{d_0}\right)^2$, where $d_r$ is much less than $d_0$ in practical scenarios.}
\begin{align}
   G^{(ff)}(\bm{\beta}) \!\approx\! \frac{ \lambda}{4\pi}\!\! \sum_{m_x=-\frac{M_x-1}{2}}^{\frac{M_x-1}{2}}\sum_{m_y=-\frac{M_y-1}{2}}^{\frac{M_y-1}{2}}\!&\frac{\sqrt{F_{m_xm_y}}}
    {\left(d_0-d_r\left(m_x \alpha_{1}+m_y\alpha_{2}\right)+\frac{d_r^2\left(m_x^2+m_y^2\right)}{2d_0}\right)} \nonumber\\
    \times  & e^{-j\left(\!k_0\left(d_0-d_r\left(m_x \alpha_{1}+m_y\alpha_{2}\right)+\frac{d_r^2\left(m_x^2+m_y^2\right)}{2d_0}
    \right)\!-\beta_{m_xm_y}\!\right)}.
    \label{endtoend_approx1}
\end{align}
When $\frac{d_r^2\left(m_x^2+m_y^2\right)}{2d_0}$  in the argument of the exponential term in (\ref{endtoend_approx1}) is much smaller than $2\pi$, we can neglect its impact.
Neglecting this term in (\ref{endtoend_approx1}) leads to the following maximum phase error 
\begin{align}
    \underset{m_x,m_y}{\text{Maximize}} \,k_0 \left( \frac{d_r^2\left(m_x^2+m_y^2\right)}{2d_0} \right)= \frac{k_0 (\frac{L_x^2}{4}+\frac{L_y^2}{4})}{2d_0}= \frac{k_0 D^2}{8d_0},
\end{align}
where $D_\text{}$ is the diagonal of the aperture.  Assuming a maximum phase error of $\frac{\pi}{8}$, the far-field region of the BS is obtained as
\begin{align}
    \frac{2\pi}{\lambda} \times \frac{D_{\text{}}^2}{8d_0} \leq \frac{\pi}{8},
\end{align}
which leads to
\begin{align}
    d_0 \geq \frac{2D_{\text{}}^2}{\lambda},
    \label{farcondition}
\end{align}
where $d_F = \frac{2D_{\text{}}^2}{\lambda}$ is referred to as the Fraunhofer distance \cite{antennatheory}. 
Neglecting $\frac{d_r^2\left(m_x^2+m_y^2\right)}{2d_0}$ in the argument of the exponential term in (\ref{endtoend_approx1}), we have
\begin{align}
    &G^{(ff)}(\bm{\beta}) \approx \frac{ \lambda e^{-jk_0d_0}}{4\pi}  \sum_{m_x=-\frac{M_x-1}{2}}^{\frac{M_x-1}{2}}\sum_{m_y=-\frac{M_y-1}{2}}^{\frac{M_y-1}{2}}\frac{\sqrt{F_{m_xm_y}} e^{j\left(k_0d_r\left(m_x \alpha_{1}+m_y\alpha_{2}
    \right)+\beta_{m_xm_y}\right)}}
    {\left(d_0-d_r\left(m_x \alpha_{1}+m_y\alpha_{2}\right)\right)}.
    \label{endtoend_approx2}
\end{align}
According to (\ref{farcondition}), since the size of the aperture is much smaller than $d_0$, we can approximate $\frac{1}{\left(d_0-d_r\left(m_x \alpha_{1}+m_y\alpha_{2}\right)\right)}$ with $\frac{1}{d_0}$. In addition, since the radiation power pattern of all elements in the far-field region of the BS is the same, we can write $F_{m_xm_y}=F$. Therefore, the BS-user channel in the far-field of the BS can be written as
\begin{align}
    &G^{(ff)}(\bm{\beta}) \approx \frac{ \sqrt{F}\lambda e^{-jk_0d_0}}{4\pi d_0} \sum_{m_x=-\frac{M_x-1}{2}}^{\frac{M_x-1}{2}}\sum_{m_y=-\frac{M_y-1}{2}}^{\frac{M_y-1}{2}}e^{j\left(k_0d_r\left(m_x \alpha_{1}+m_y\alpha_{2}
    \right)+\beta_{m_xm_y}\right)}.
    \label{endtoend_approx3}
\end{align}
It can be observed from (\ref{endtoend_approx3}) that $G^{(ff)}(\bm{\beta})$ attains its maximum value when we set
\begin{align}
    \beta_{m_xm_y}=-\text{mod}\left( k_0 d_r \left(m_x \alpha_1 + m_y \alpha_2 \right),2\pi\right), \:\forall \, m_x,m_y.
    \label{optimalphase_far}
\end{align}
From (\ref{optimalphase_far}), we can conclude that in order to obtain the optimal phase shift for maximizing the effective BS-user channel in the far-field region of the BS, we need to estimate only two parameters, $\alpha_1$ and $\alpha_2$.

Since we do not have the values of $\alpha_1$ and $\alpha_2$, we consider a general case to see what would be the BS-user channel in the far-field of the BS if we replace $\alpha_1$ and $\alpha_2$ with any other values. In the following lemma, we apply a linear phase shift to each phase-shifting element to obtain a closed-form expression for the effective BS-user channel in the far-field of the BS.

\begin{lemma}
\label{lemma1}
 \normalfont
If we apply the following linear phase shift across the aperture elements, i.e., to the $(m_x,m_y)$-th element
\begin{align}
    \beta_{m_xm_y} = -\text{mod}\Big( k_0 d_r \big(m_x \beta_1 + m_y \beta_2 \big),2\pi\Big), \:\forall \, m_x,m_y,
    \label{phase_shift_far}
\end{align}
then the BS-user channel in the far-field of the BS is a function of $\beta_1$ and $\beta_2$, and can be obtained as 
\begin{align}
    &G^{(ff)}(\beta_1,\beta_2) \approx \left(\frac{\sqrt{F}\lambda e^{-jk_0d_0}}{4\pi d_0}\right)  \frac{\sin{\Big(\frac{k_0L_x}{2}\left(\alpha_1-\beta_1\right)\!\Big)}}{\sin{\left(\frac{k_0d_r}{2}\left(\alpha_1-\beta_1\right)\right)}} \times \frac{\sin{\left(\frac{k_0L_y}{2}\left(\alpha_2-\beta_2\right)\right)}}{\sin{\left(\frac{k_0d_r}{2}\left(\alpha_2-\beta_2\right)\right)}}.
    \label{ITS-user_1}
\end{align}
\end{lemma}
\begin{proof}
 The proof follows similar steps as [42, and 47, Ch. 6] and is provided in  Appendix A for completeness.
\end{proof}
For the extremely sub-wavelength elements ($d_r \rightarrow 0$), the aperture acts as a continuous surface  \cite{huang2020holographic}. Assuming $d_r \rightarrow 0$, we can approximate $\sin{\left(\frac{k_0d_r}{2}\left(\alpha_1-\beta_1\right)\right)}$ and $\sin{\left(\frac{k_0d_r}{2}\left(\alpha_2-\beta_2\right)\right)}$ in the denominator of (\ref{ITS-user_1}) with $\frac{k_0d_r}{2}\left(\alpha_1-\beta_1\right)$ and $\frac{k_0d_r}{2}\left(\alpha_2-\beta_2\right)$, respectively. Then, the BS-user channel in the far-field of the BS in (\ref{ITS-user_1}) can be written as
\begin{align}
    G^{(ff)}(\beta_1,\beta_2) \approx \left(\frac{ \sqrt{F}\lambda e^{-jk_0d_0}}{4\pi d_0}\right)\! M_x M_y \sinc{\left(\frac{k_0 L_x}{2}\left(\alpha_1\!-\!\beta_1\right)\right)} \sinc{\left(\frac{k_0 L_y}{2}\left(\alpha_2\!-\!\beta_2\right)\right)},
    \label{ITS-user_2}
\end{align}
where $\sinc(x)=\frac{\sin(x)}{x}$.
According to (\ref{ITS-user_2}), the absolute value of the BS-user channel would be maximized when the sinc functions attain their maximum value, which occurs for the first and second sinc functions  when $\beta_1$ and $\beta_2$ are set as $\beta_1 = \alpha_1$ and $\beta_2 = \alpha_2$, respectively. 
Therefore, when we use the far-field channel model, we only need to estimate two parameters, $\alpha_1$ and $\alpha_2$.
 Compared with the optimal phase shift in (\ref{nearphase_1}), we observe that the number of parameters needed to be estimated decreases from three to two. 
 Most importantly, we have avoided the estimation of $d_0$, which was the problematic parameter for estimation. In Sec. \ref{absence}, we exploit the properties of these two sinc functions to propose the channel estimation scheme.

\subsection{Near-Field Channel Model}
\label{nearfield}
\begin{figure}
\centering
    \includegraphics[scale=0.6]{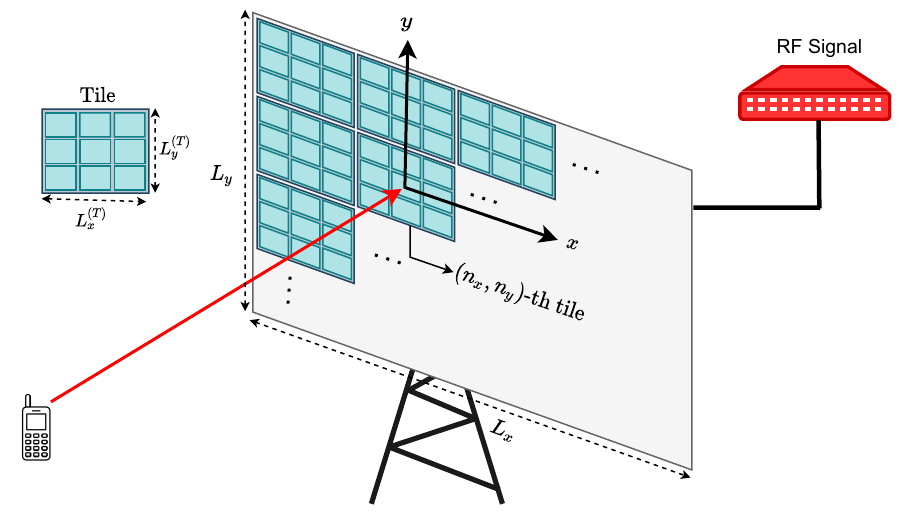}
\caption{\small{The aperture is partitioned into tiles such that the far-field condition holds for each tile.
}}
\label{fig4}
\end{figure}
In this subsection, we will present a model for the effective BS-user channel in (\ref{endtoend_exact_ITS}) that holds in the near-field region of the aperture. According to (\ref{farcondition}), the near-field region of the BS is given by\footnote{We assume that $d_0$ is larger than the Fresnel distance, denoted by $d_n = \sqrt[3]{\frac{D^4}{8\lambda}}$, so the system
does not operate in the reactive near-field of the aperture \cite{antennatheory}.}
\begin{align}
 d_0\leq \frac{2D^2}{\lambda}.
\end{align}

Similar to \cite{najafi2020physics}, we partition the aperture into tiles which are small enough to satisfy the far-field condition in (\ref{farcondition}). Different tiles can jointly configure their phase-shifting elements to maximize the channel through each tile to the user. 
As shown in Fig. \ref{fig4}, we assume that the aperture is partitioned into $N_x \times N_y$ tiles of size $L_x^{(\text{T})} \times L_y^{(\text{T})}$, where $L_x^{(\text{T})}=L_x/N_x$ and $L_y^{(\text{T})}=L_y/N_y$ are the width and the length of the tile, respectively.
The total number of phase-shifting elements of each tile is given by $M^{(\text{T})}=M^{(\text{T})}_x \times M^{(\text{T})}_y$, where $M^{(\text{T})}_x=L^{(\text{T})}_x/d_r$ and $M^{(\text{T})}_y=L^{(\text{T})}_y/d_r$.
Assuming $N_x$ and $N_y$ are odd numbers, the position of the center of the $(n_x,n_y)$-th tile is given by $(x,y)=(n_xL_x^{(\text{T})},n_yL_y^{(\text{T})})$ for $n_x=-\frac{N_x-1}{2}, \cdots, \frac{N_x-1}{2}$ and $n_y=-\frac{N_y-1}{2}, \cdots, \frac{N_y-1}{2}$. 


According to the far-field condition in (\ref{farcondition}), we can use the far-field channel model between a given tile of the aperture and the user if 
the following condition holds
\begin{align}
 d_0\geq \frac{2D_{\text{T}}^2}{\lambda},
 \label{near_field_tile}
\end{align}
 where $D_{\text{T}}$ is the diagonal of the tile. Substituting $D_{\text{T}}=\sqrt{\left(\frac{L_x}{N_x}\right)^2+\left(\frac{L_y}{N_y}\right)^2}$ into (\ref{near_field_tile}), we can obtain $N_x$ and $N_y$ from the following inequality
\begin{align} 
 d_0 \geq \frac{2\left(\left(\frac{L_x}{N_x}\right)^2+\left(\frac{L_y}{N_y}\right)^2\right)}{\lambda}.   
\end{align}

We first focus on the channel between the $(n_x,n_y)$-th tile of the aperture and the user. 
Without loss of generality, we assume that the origin of the coordinate system is placed at the center of the $(n_x,n_y)$-th tile.
Let $\theta_{n_xn_y}$ and $\phi_{n_xn_y}$ denote the elevation and azimuth angles of the impinging wave from the user to the center of the $(n_x,n_y)$-th tile of the aperture.
Similar to (\ref{alpha_1}) and (\ref{alpha_2}), let us define $\alpha_{1_{n_xn_y}}= \sin{\left(\theta_{n_xn_y}\right)} \cos{\left(\phi_{n_xn_y}\right)}$, $\alpha_{2_{n_xn_y}}= \sin{\left(\theta_{n_xn_y}\right)} \sin{\left(\phi_{n_xn_y}\right)}$,  for the $(n_x,n_y)$-th tile at the aperture. 
Moreover, let $d_{n_xn_y}^{(\text{T})}$ denote the distance between the center of the $(n_x,n_y)$-th tile and the user, see Fig. \ref{fig4}.
 Assuming $M^{(\text{T})}_x$ and $M^{(\text{T})}_y$ are odd numbers, the position of the $(m_x,m_y)$-th phase-shifting element of the $(n_x,n_y)$-th tile is given by $(x,y)=(m_xd_r,m_yd_r)$ for $m_x=-\frac{M^{(\text{T})}_x-1}{2},\cdots, \frac{M^{(\text{T})}_x-1}{2}$ and $m_y=-\frac{M^{(\text{T})}_y-1}{2},\cdots, \frac{M^{(\text{T})}_y-1}{2}$.
Similar to Lemma \ref{lemma1}, if we apply the following linear phase shift to the $(m_x,m_y)$-th phase-shifting element of the $(n_x,n_y)$-th tile
\begin{align}
    \beta^{n_xn_y}_{m_xm_y} = -\text{mod}\Big( k_0 \big(m_x\beta_{1_{n_xn_y}} +m_y  \beta_{2_{n_xn_y}}  \big) +\beta_{0_{n_xn_y}},2\pi\Big)\:\forall \, m_x,m_y ,
    \label{phase_shift-cont}
\end{align}
then, the channel between the $(n_x,n_y)$-th tile and the user, denoted by $G_{n_xn_y}(\beta_{0_{n_xn_y}}, \beta_{1_{n_xn_y}},$ $\beta_{2_{n_xn_y}})$, can be obtained as
\begin{align}
  &G_{n_xn_y}\left(\beta_{0_{n_xn_y}},\beta_{1_{n_xn_y}},\beta_{2_{n_xn_y}}\right) \approx \left(\frac{ \sqrt{F_{n_xn_y}}\lambda e^{-jk_0 d_{n_xn_y}^{(\text{T})}}}{4\pi d_{n_xn_y}^{(\text{T})}}\right)
    M_x^{(\text{T})} M_y^{(\text{T})} \times\nonumber \\ 
    & \hspace{3cm}  \sinc\!{\left(\!\frac{k_0 L_x^{(\text{T})}}{2}\!\left(\alpha_{1_{n_xn_y}}\!\!-\!\beta_{1_{n_xn_y}}\right)\!\!\right)}\! \sinc\!{\left(\!\frac{k_0 L_y^{(\text{T})}}{2}\!\left(\alpha_{2_{n_xn_y}}\!\!-\!\beta_{2_{n_xn_y}}\right)\!\!\right)} e^{-j\beta_{0_{n_xn_y}}},
    \label{channel-tile}
\end{align}
where $\beta_{0_{n_xn_y}}$ determines the phase of the channel \cite{najafi2020physics} and $F_{n_xn_y}$ is the radiation power pattern of the phase-shifting elements of $(n_x,n_y)$-th tile. 

For ease of presentation, let us define the phase shift parameters of the aperture for all tiles as $\left(\bm{\beta_0},\bm{\beta_1},\bm{\beta_2}\right)= \left\{\left(\beta_{0_{n_xn_y}},\beta_{1_{n_xn_y}},\beta_{2_{n_xn_y}}\right);\, \forall n_x, n_y \right\}$. The BS-user channel in the near-field region of the BS, denoted by $G^{(nf)}\left(\bm{\beta_0},\bm{\beta_1},\bm{\beta_2}\right)$, is the superposition of the channels through individual tiles and is obtained as
\begin{align} 
   \hspace{-0.2cm} G^{(nf)}\!\left(\bm{\beta_0},\bm{\beta_1},\!\bm{\beta_2}\right) \approx&\! \sum_{n_x=-\frac{N_x-1}{2}}^{\frac{N_x-1}{2}}\!\sum_{n_y=-\frac{N_y-1}{2}}^{\frac{N_y-1}{2}}\!H_{n_xn_y}\!\left(\beta_{0_{n_xn_y}},\beta_{1_{n_xn_y}},\beta_{2_{n_xn_y}}\right)\nonumber\\
    \approx& \sum_{n_x=-\frac{N_x-1}{2}}^{\frac{N_x-1}{2}}\!\sum_{n_y=-\frac{N_y-1}{2}}^{\frac{N_y-1}{2}}  \left(\frac{ \sqrt{F}\lambda e^{-jk_0 d_{n_xn_y}^{(\text{T})}}}{4\pi d_{n_xn_y}^{(\text{T})}}\right)
    M_x^{(\text{T})} M_y^{(\text{T})} \nonumber \\ 
     \times & \sinc\!{\left(\!\frac{k_0 L_x^{(\text{T})}}{2}\!\left(\alpha_{1_{n_xn_y}}\!\!-\!\beta_{1_{n_xn_y}}\right)\!\!\right)}\! \sinc\!{\left(\!\frac{k_0 L_y^{(\text{T})}}{2}\!\left(\alpha_{2_{n_xn_y}}\!\!-\!\beta_{2_{n_xn_y}}\right)\!\!\right)} e^{-j\beta_{0_{n_xn_y}}}.
    \label{channel-near-field}
\end{align}
According to (\ref{channel-near-field}), in order to maximize the BS-user channel in the near-field region, first we need to maximize the channel between the $(n_x,n_y)$-th tile of the BS and the user in (\ref{channel-tile}), i.e.,  $G_{n_xn_y}\left(\beta_{0_{n_xn_y}},\beta_{1_{n_xn_y}},\beta_{2_{n_xn_y}}\right)$, $\forall n_x, n_y$. In order to maximize $G_{n_xn_y}\left(\beta_{0_{n_xn_y}},\beta_{1_{n_xn_y}},\beta_{2_{n_xn_y}}\right)$ for the $(n_x,n_y)$-th tile, we need to estimate two parameters, $\alpha_{1_{n_xn_y}}$ and $\alpha_{2_{n_xn_y}}$. Finally, $\beta_{0_{n_xn_y}}$ $\forall n_x, n_y$ in (\ref{channel-near-field}) can be obtained such that the channels through the individual tiles have the same phase. 
Since the number of tiles is $N_x \times N_y$, we need to estimate $2N_x \times N_y$ parameters for the near-filed region of the BS and we have avoided the problematic estimation of $d_0$. Note that these parameters are dependent to each other and we use this dependency in the proposed channel estimation scheme.

\emph{Remark 1:} In the channel estimation scheme based on the location of the user, three parameters are needed to be estimated, i.e., $d_0$, $\alpha_1$, and $\alpha_2$. Substituting these parameters in (\ref{nearphase_1}), the optimal phase shifts at the aperture can be obtained. Compered to these schemes, the idea of partitioning the aperture to the tiles increases the number of unknown parameters from three to $2 N_x \times N_y$. However, the beam created by the entire aperture has a narrower beamwidth compared to each tile. Therefore, an error in the estimation of the location has potentially a larger negative impact on the performance of the location-based schemes especially for large apertures. On the other hand, it is clear from (\ref{channel-near-field}) that the idea of partitioning ensures constructive superposition by $\beta_{0_{n_xn_y}}$. In addition, since the beamwidth of each tile is larger than the entire aperture, an error in the estimation of $\alpha_{1_{n_xn_y}}$ and $\alpha_{2_{n_xn_y}}$ has a lower negative impact on the performance.

\section{Proposed Channel Estimation}
\label{absence}
Now that we have identified the parameters that need to be estimated in the far-field and near-field regions, we move towards proposing our scheme for their estimation by exploiting the specific structure of the radiated beam in the dominated LoS channel model. 
Specifically, to obtain some intuition, we first propose the channel estimation scheme under the assumption that there is no noise in the system. Once we obtain the intuitions, then we propose an iterative algorithm for the channel estimation in the presence of the noise. Since we have decomposed the near-field as a superposition of the far-field, we only present the channel estimation scheme for the far-field scenario.

\subsection{Channel Estimation in the Absence of Noise}
\label{passive estimation}
In this subsection, we propose the channel estimation scheme in the absence of noise to provide some intuitions, which is a widely adopted approach in literature \cite{sanchez2021gridless}, \cite{wang2021joint}, \cite{wang2020channel}. The study of the noise-less case (albeit not practical) reveals the key features of the radiated pattern exploited here for channel estimation and sets the basis for the proposed estimator in the noisy case. 
In the channel estimation procedure, the user sends a pilot signal $x_p = \sqrt{P_p}$ to the BS, where $P_p$ is the pilot transmit power. In the absence of noise, the received signal at the BS is given by
\begin{align}
    y(\beta_1,\beta_2) = \sqrt{P_p}\times  G^{(ff)}(\beta_1,\beta_2).
    \label{received_absence}
\end{align}
 Substituting (\ref{ITS-user_2}) into (\ref{received_absence}) and assuming $L_x = K_x \lambda$ and $L_y = K_y \lambda$, where $K_x,K_y \in \mathbb{N}$ are integer numbers, the absolute value of the received signal at the BS is given by
\begin{align}
  \Big|  y(\beta_1,\beta_2)\Big| \approx \sqrt{P_p} \left(\frac{ \sqrt{F}\lambda }{4\pi d_0}\right) \bigg| M_x M_y \sinc{\left(K_x \pi \left(\alpha_1-\beta_1\right)\right)} \sinc{\left(K_y\pi \left(\alpha_2-\beta_2\right)\right)}\bigg|.
    \label{received_pilot}
\end{align}

In the following, we show that $\alpha_1$ and $\alpha_2$ can be estimated using five pilots sent by the user. Before sending each pilot, the BS applies a new phase shift to the elements by changing $\beta_1$ and $\beta_2$ in (\ref{phase_shift_far}). For the first pilot, the BS sets $\beta_1=\hat{\beta_1}$ and $\beta_2=\hat{\beta_2}$, where $\hat{\beta_1}$ and $\hat{\beta_2}$ are two random numbers in the range of -1 to 1. Note that, according to (\ref{alpha_1}) and (\ref{alpha_2}), $\alpha_1$ and $\alpha_2$ are in the range of -1 to 1. Then, the absolute value of the received signal at the BS due to the first pilot signal is given by 
\begin{align}
   \left| y(\hat{\beta_1},\hat{\beta_2})\right| \approx \sqrt{P_p} \left(\frac{ \sqrt{F}\lambda }{4\pi d_0}\right) \bigg| M_x M_y \sinc{\left(K_x\pi \left(\alpha_1-\hat{\beta_1}\right)\right)} \sinc{\left(K_y\pi \left(\alpha_2-\hat{\beta_2}\right)\right)}\bigg|.
    \label{received_pilot_1}
\end{align}
For the second pilot, the BS sets $\beta_1=\hat{\beta_1}+v$ and $\beta_2=\hat{\beta_2}$, where the value of $v$ will be discussed later. Therefore, the absolute value of the received signal at the BS is
\begin{align}
   \left| y(\hat{\beta_1}\!+\!v,\hat{\beta_2})\right| \!\approx \! \sqrt{P_p} \left(\frac{ \sqrt{F}\lambda }{4\pi d_0}\right)\! \bigg| M_x M_y \sinc{\!\left(K_x\pi \!\left(\alpha_1\!-\!\hat{\beta_1}\!-\!v\right)\!\right)} \sinc{\!\left(K_y\pi \!\left(\alpha_2\!-\!\hat{\beta_2}\right)\!\right)}\bigg|.
    \label{received_pilot_2}
\end{align}
If the BS divides (\ref{received_pilot_1}) by (\ref{received_pilot_2}), it will obtain
\begin{align}
    \frac{ \left|   y\left(\hat{\beta_1},\hat{\beta_2}\right)\right|}{ \left|   y\left(\hat{\beta_1}+v,\hat{\beta_2}\right)\right|} \approx \frac{ \left|\sinc \left( K_x\pi\left(\alpha_1-\hat{\beta_1}\right)\right)\right|}{\left| \sinc \left(K_x\pi\left(\alpha_1-\hat{\beta_1}-v\right)\right)\right|} = \frac{\left|\frac{\sin{\left( K_x\pi\left(\alpha_1-\hat{\beta_1}\right)\right)}}{ K_x\pi\left(\alpha_1-\hat{\beta_1}\right)}\right|}{\left|\frac{\sin{\left( K_x\pi\left(\alpha_1-\hat{\beta_1}-v\right)\right)}}{\ K_x\pi\left(\alpha_1-\hat{\beta_1}-v\right)}\right|}.
    \label{devision}
\end{align}
If $v$ is selected such that $K_x v\in \mathbb{N}$, we have $\sin{( K_x\pi(\alpha_1-\hat{\beta_1}-v))}=\sin{( K_x\pi(\alpha_1-\hat{\beta_1}))}$. Then, (\ref{devision}) can be simplified to 
\begin{align}
    \frac{ \left|   y\left(\hat{\beta_1},\hat{\beta_2}\right)\right|}{ \left|   y\left(\hat{\beta_1}+v,\hat{\beta_2}\right)\right|} \approx \left|\frac{\alpha_1-\hat{\beta_1}-v}{\alpha_1-\hat{\beta_1}}\right|.
    \label{alpha_11_absence noise}
\end{align}
From equation (\ref{alpha_11_absence noise}), two solutions for $\alpha_1$, denoted by $\alpha_1^{(1)}$ and $\alpha_1^{(2)}$, can be obtained as
\begin{align}
   & \alpha_1^{(1)/(2)} = \hat{\beta_1} + \frac{\left|   y\left(\hat{\beta_1}+v,\hat{\beta_2}\right)\right|}{\left|   y\left(\hat{\beta_1}+v,\hat{\beta_2}\right)\right| \pm \left|   y\left(\hat{\beta_1},\hat{\beta_2}\right)\right|}v. 
    \label{alpha1_1_12_absence noise}
\end{align}
In order to identify the correct solution for $\alpha_1$, the third pilot signal should be sent from the user. For the third pilot signal, the BS sets $\beta_1=\hat{\beta_1}-v$ and $\beta_2=\hat{\beta_2}$. Using the received signal of the first and third pilot signals, two other solutions for $\alpha_1$, denoted by $\alpha_1^{(3)}$ and $\alpha_1^{(4)}$, can be similarly obtained as
\begin{align}
   & \alpha_1^{(3)/(4)} = \hat{\beta_1} + \frac{\left|   y\left(\hat{\beta_1}-v,\hat{\beta_2}\right)\right|}{-\left|   y\left(\hat{\beta_1}-v,\hat{\beta_2}\right)\right| \pm \left|   y\left(\hat{\beta_1},\hat{\beta_2}\right)\right|}v.
    \label{alpha1_1_34_absence noise}
\end{align}
One of the solutions in (\ref{alpha1_1_12_absence noise}) is approximately the same as one of the solutions in (\ref{alpha1_1_34_absence noise}). Therefore, using (\ref{alpha1_1_12_absence noise}) and (\ref{alpha1_1_34_absence noise}), the correct solution for $\alpha_1$ can be obtained as 
\begin{align}
\alpha_1\approx \left\{\frac{\alpha_1^{(i)}+\alpha_1^{(j)}}{2}\mathrel{\Big|} \underset{i,j}{\text{min}}|\alpha_1^{(i)}-\alpha_1^{(j)}|; i\in\{1,2\}, j \in \{3,4\}\right\}.
     \label{alpha_hat1}
\end{align}

In order to obtain $\alpha_2$, two more pilot signals are needed to be sent by the user. For these two pilot signals, the BS sets two different phase shifts as  $(\beta_1,\beta_2)=\{ (\hat{\beta_1},\hat{\beta_2}+w), (\hat{\beta_1},\hat{\beta_2}-w)\}$, where $w$ is selected such that $K_y w\in \mathbb{N}$. Then, similar to $\alpha_1$, we have the following solutions for $\alpha_2$ 
 \begin{align}
   & \alpha_2^{(1)/(2)} = \hat{\beta_2} + \frac{\left|   y\left(\hat{\beta_1},\hat{\beta_2}+w\right)\right|}{\left|   y\left(\hat{\beta_1},\hat{\beta_2}+w\right)\right| \pm \left|   y\left(\hat{\beta_1},\hat{\beta_2}\right)\right|}w, \label{alpha1_2_12_absense noise}\\
   & \alpha_2^{(3)/(4)} =\hat{\beta_2} +  \frac{\left|   y\left(\hat{\beta_1},\hat{\beta_2}-w\right)\right|}{-\left|   y\left(\hat{\beta_1},\hat{\beta_2}+w\right)\right| \pm \left|   y\left(\hat{\beta_1},\hat{\beta_2}\right)\right|}w.
    \label{alpha1_2_34_absence noise}
\end{align}
Then, using (\ref{alpha1_2_12_absense noise}) and (\ref{alpha1_2_34_absence noise}), the correct solution for $\alpha_2$ can be obtained as 
\begin{align}
\alpha_2 \approx \left\{\frac{\alpha_2^{(i)}+\alpha_2^{(j)}}{2}\mathrel{\Big|} \underset{i,j}{\text{min}}|\alpha_2^{(i)}-\alpha_2^{(j)}|; i\in\{1,2\}, j \in \{3,4\}\right\}.
     \label{alpha_hat2}
\end{align}
From (\ref{alpha_hat1}) and (\ref{alpha_hat2}), we can conclude that when the channel is dominated by LoS path, we can find the unknown parameters without any ambiguity using five pilot signals in the absence of noise\footnote{In Sec. \ref{num_result}, we show that at a high signal-to-noise ratio (SNR) regime, the proposed scheme requires only five pilots to estimate the unknown channel parameters.}.

\emph{Remark 2:} The random choices for $\hat{\beta_1}$ and $\hat{\beta_2}$ may lead to the null points of the sinc functions in (\ref{received_pilot_1}), which occurs at $\hat{\beta_1}=\alpha_1 \pm \frac{q}{K_x} $ and $\hat{\beta_2}=\alpha_2 \pm \frac{q}{K_y} $, where $q \in \mathbb{N}$. For these unfortunate initial values,  the received signals at the BS are zero, and hence we cannot estimate $\alpha_1$ and $\alpha_2$. Since the initial values are chosen randomly, the probability of occurrence at the null points of the sinc functions is zero. However, if this happens, we change the initial values from $\hat{\beta_1}$ and $\hat{\beta_2}$ to $\hat{\beta_1}+\frac{1}{2K_x}$ and $\hat{\beta_2}+\frac{1}{2K_y}$, respectively, to move from the nulls to their closest peaks. 

\subsection{Channel Estimation in the Presence of Noise}
\label{presence}
The proposed scheme in the previous subsection perfectly estimates the channel parameters for the ideal case when there is no noise at the receiver. However, that case is only for intuition purposes since noise is always unavoidable in practice. In the presence of noise at the receiver, the estimated values in (\ref{alpha_hat1}) and (\ref{alpha_hat2}) have errors.  
In this section, we utilize the estimation scheme introduced in Sec. \ref{passive estimation} to propose an iterative algorithm to decrease the channel estimation error due to the noise. 

In the presence of noise, if the user sends the pilot $x_p=\sqrt{P_p}$ to the BS, the received signal at the BS, denoted by $\hat{y}$, is given by
\begin{align}
    \hat{y}(\beta_1,\beta_2) = \sqrt{P_p}\times G^{(ff)}(\beta_1,\beta_2)+n,
    \label{received_presence}
\end{align}
where $n$ denotes the additive white Gaussian noise (AWGN) at the BS. 
\begin{algorithm}[t]
\scriptsize
\caption{\small{Finding the estimated values for $\alpha_1$ and $\alpha_2$}.
}
\label{algorithm}
\begin{algorithmic}[1]
\State Initialize $k=0$, set    $\left(\hat{\beta}_1^{(0)},\hat{\beta}_2^{(0)}\right)$ to a random pair between $[-1,1]$.
\While {convergence==False}
\State $k=k+1$ and obtain $\hat{\alpha}_1$ and $\hat{\alpha}_2$ by (\ref{final_alpha_1_presence}) and (\ref{final_alpha_2_presence}).
\State update $\left(\hat{\beta}_1^{(k)},\hat{\beta}_2^{(k)}\right)$ by $(\hat{\alpha}_1,\hat{\alpha}_2)$.
\If{$    \left(\hat{\beta}_1^{(k)}-\hat{\beta}_1^{(k-1)}\right)^2 +    \left(\hat{\beta}_2^{(k)}-\hat{\beta}_2^{(k-1)}\right)^2 < \delta,$ or $k > K_{\text{max}}$,} 
\quad convergence = True
\EndIf
\EndWhile\\
  \Return $\alpha_1=\hat{\beta}_1^{(k)}$, $\alpha_2=\hat{\beta}_2^{(k)}$.
\end{algorithmic}
\end{algorithm}

The proposed iterative algorithm for noisy channel estimation is presented in Algorithm \ref{algorithm}, in which index $k$ is used to denote the $k$-th iteration. This algorithm works as follows. As the starting point for the iterative algorithm, we choose  $\left(\hat{\beta_1}^{\left(0\right)},\hat{\beta_2}^{\left(0\right)}\right)$ as a random pair between $[-1,1]$. 
In each iteration, the user sends five pilots to the BS. Then, the BS estimates $\alpha_1$ and $\alpha_2$ based on the received pilots, as follows.
In the $k$-th iteration, similar to (\ref{alpha1_1_12_absence noise}) and (\ref{alpha1_1_34_absence noise}), four solutions for $\alpha_1$ can be obtained as 
\begin{align}
  \alpha_1^{(1)/(2)} =\hat{\beta_1}^{\left(k-1\right)} + \frac{\left|   \hat{y}\left(\hat{\beta_1}^{\left(k-1\right)}+v,\hat{\beta_2}^{\left(k-1\right)}\right)\right|}{\left|   \hat{y}\left(\hat{\beta_1}^{\left(k-1\right)}+v,\hat{\beta_2}^{\left(k-1\right)}\right)\right| \pm \left|   \hat{y}\left(\hat{\beta_1}^{\left(k-1\right)},\hat{\beta_2}^{\left(k-1\right)}\right)\right|}v, 
    \label{alpha1_1_12_presence noise}
\end{align}
\begin{align}
    \alpha_1^{(3)/(4)} =\hat{\beta_1}^{\left(k-1\right)} + \frac{\left|   \hat{y}\left(\hat{\beta_1}^{\left(k-1\right)}-v,\hat{\beta_2}^{\left(k-1\right)}\right)\right|}{-\left|   \hat{y}\left(\hat{\beta_1}^{\left(k-1\right)}-v,\hat{\beta_2}^{\left(k-1\right)}\right)\right| \pm \left|   \hat{y}\left(\hat{\beta_1}^{\left(k-1\right)},\hat{\beta_2}^{\left(k-1\right)}\right)\right|}v.
    \label{alpha1_1_34_presence noise}
\end{align}
Similar to (\ref{alpha_hat1}), we obtain $\alpha_1$ as the average of two answers with the minimum difference. 
 \begin{align}
   \hat{\alpha}_1 = \left\{\frac{\alpha_1^{(i)}+\alpha_1^{(j)}}{2}\mathrel{\Big|} \underset{i,j}{\text{min}}|\alpha_1^{(i)}-\alpha_1^{(j)}|; i\in\{1,2\}, j \in \{3,4\}\right\}.
     \label{final_alpha_1_presence}
\end{align}

Similarly, four answers for $\alpha_2$ can be obtained as 
\begin{align}
  & \alpha_2^{(1)/(2)} =\hat{\beta_2}^{\left(k-1\right)} + \frac{\left|   \hat{y}\left(\hat{\beta_1}^{\left(k-1\right)},\hat{\beta_2}^{\left(k-1\right)}+w\right)\right|}{\left|   \hat{y}\left(\hat{\beta_1}^{\left(k-1\right)},\hat{\beta_2}^{\left(k-1\right)}+w\right)\right| \pm \left|   \hat{y}\left(\hat{\beta_1}^{\left(k-1\right)},\hat{\beta_2}^{\left(k-1\right)}\right)\right|}w, 
    \label{alpha1_2_12_presence noise}\\
   & \alpha_2^{(3)/(4)} =\hat{\beta_2}^{\left(k-1\right)} + \frac{\left|   \hat{y}\left(\hat{\beta_1}^{\left(k-1\right)},\hat{\beta_2}^{\left(k-1\right)}-w\right)\right|}{-\left|   \hat{y}\left(\hat{\beta_1}^{\left(k-1\right)},\hat{\beta_2}^{\left(k-1\right)}-w\right)\right| \pm \left|   \hat{y}\left(\hat{\beta_1}^{\left(k-1\right)},\hat{\beta_2}^{\left(k-1\right)}\right)\right|}w.
    \label{alpha1_2_34_presence noise}
\end{align}
Then, $\hat{\alpha}_2$ can be obtained as 
 \begin{align}
    \hat{\alpha}_2 = \left\{\frac{\alpha_2^{(i)}+\alpha_2^{(j)}}{2}\mathrel{\Big|} \underset{i,j}{\text{min}}|\alpha_2^{(i)}-\alpha_2^{(j)}|; i\in\{1,2\}, j \in \{3,4\}\right\}.
     \label{final_alpha_2_presence}
\end{align}
Now, the process of updating $\left(\hat{\beta_1}^{(k)},\hat{\beta_2}^{(k)}\right)$ by $\left(\hat{\alpha}_1,\hat{\alpha}_2\right)$ continues until convergence occurs. The convergence criterion is checked at the end of each iteration by observing whether the following inequality holds
\begin{align}
    \left(\hat{\beta_1}^{(k)}-\hat{\beta_1}^{(k-1)}\right)^2 +    \left(\hat{\beta_2}^{(k)}-\hat{\beta_2}^{(k-1)}\right)^2 < \delta,
    \label{creterion}
\end{align}
where $\delta$ is a small threshold. If (\ref{creterion}) holds, the iteration stops and $\left(\hat{\beta_1}^{(k)},\hat{\beta_2}^{(k)}\right)$ are adopted as the correct estimations.


\subsection{Initial values for the Iterative Algorithm}
\begin{figure}
\centering
    \includegraphics[scale=0.7]{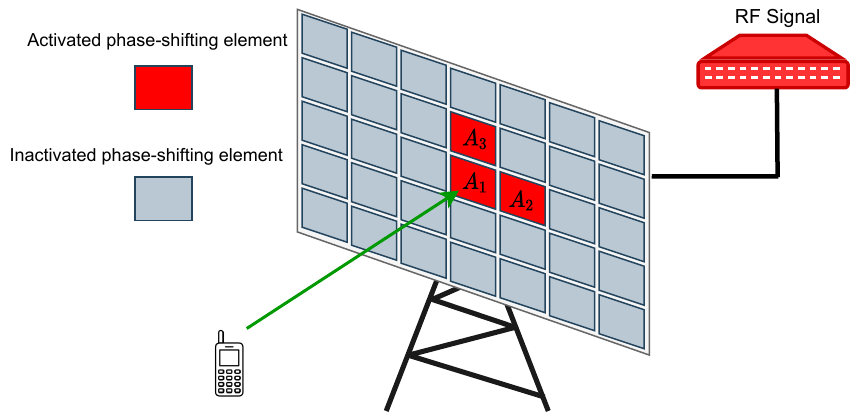}
\caption{\small{Three radiating elements placed at $(0,0)$, $\left(d_r,0\right)$, and $\left(0,d_r\right)$ are activated in three different time slot to provide initial values for the iterative algorithm.}}
\label{figure5}
\end{figure}

The performance of the proposed iterative algorithm depends on the initial values of $\hat{\beta}_1^{(0)}$ and $\hat{\beta}_2^{(0)}$.
When the initial values of the iterative algorithm are random, it means that the BS focuses the reception of the pilot signals from random directions. As a result, the received signal at the BS may be comparable to the noise, which degrades the performance of the iterative algorithm. On the other hand, when the initial values are close to $\alpha_1$ and $\alpha_2$, the iterative algorithm can accurately estimate $\alpha_1$ and $\alpha_2$.
In the following, we propose a simple method to provide the initial values of $\hat{\beta}_1^{(0)}$ and $\hat{\beta}_2^{(0)}$ that are close to $\alpha_1$ and $\alpha_2$, respectively.  

In order to provide initial values, the user sends three pilot signals in three different time slots. In each time slot, the BS activates only one phase-shifting element. In the $i$-th time slot, the BS activates the $A_i$-th element, where $A_1$, $A_2$, and $A_3$ are placed at $(0,0)$, $(d_r,0)$, and  $(0,d_r)$, respectively, see Fig. \ref{figure5}. 
 Let $d_i, i \in \{1,2,3\}$, denote the distance between the  user and the $A_i$-th activated element. The user sends the pilot $x_p = \sqrt{P_p}$ to the BS, where $P_p$ is the pilot transmit power.
 According to (\ref{ITSelement_user}), the received signal at the BS in the $i$-th time slot, denoted by $y_{i}$, is given by
\begin{align}
    y_{i} = \sqrt{P_p}\frac{\sqrt{F_i}\lambda}{4\pi d_i}e^{-jk_0d_i} + n_i, 
    \quad i\in\{1,2,3\},
    \label{system_equation}
\end{align}
where $F_i$ accounts for the power radiation pattern of the $A_i$-th activated element and the user's antenna and $n_i$ denotes the AWGN at the $i$-th time slot. Since these activated elements are close to each other, $F_i$ is the same for these elements, i.e., $F_i = F$.

From (\ref{fard_mxmy}), the distance between three activated elements and the user can be written as
\begin{align}
    &d_1 = d_0, \label{50}\\
    &d_2 \approx d_0 -d_r \alpha_{1},\label{51}\\
    &d_3 \approx d_0 -d_r \alpha_{2}.
    \label{52}
\end{align}
Substituting (\ref{50})-(\ref{52}) into (\ref{system_equation}) and using the phase of the received signal at the activated elements, we can estimate $\alpha_{1}$ and $\alpha_{2}$ as follows
\begin{align}
   &\hat{\alpha}_1 = \frac{\measuredangle y_{2}-\measuredangle y_{1}}{k_0 d_r}, \label{activealpha_1t}\\
   &\hat{\alpha}_2 = \frac{\measuredangle y_{3}-\measuredangle y_{1}}{k_0 d_r}, 
   \label{activealpha_2t}
\end{align}
where $\measuredangle y_i$ denotes the phase of $y_i$.
It is worth noting that due to the noise, $\hat{\alpha}_1$ and $\hat{\alpha}_2$ are not accurate enough. Since the BS with infinitely many antennas provides super-resolution beamforming, a small error in the estimation values can degrade the performance of the system during data transmission. Therefore, these estimated values can only provide good initial values for the iterative algorithm.

\subsection{Extension to Near-Field Regime}
According to the far-field condition in (\ref{farcondition}), the user lies in the far-field region if the distance between the BS and the user is larger than the Fraunhofer distance $d_F = \frac{2 D^2}{\lambda}$. 
When the distance between the BS and the user is shorter than $d_F$, the user lies in the near-field region of the aperture. In conventional wireless communication systems, due to the small size of apertures and using frequencies with centimeter wavelength, the Fraunhofer distance is usually several meters, for which far-field assumption typically holds in practice. On the other hand, in future mmWave and THz communications, due to the significant increase in the electric size of apertures and operation at higher frequencies, the near-field region can be up to several hundreds of meters. For example, for an aperture of size $L_x = L_y = 0.5 ~\rm{m}$ at carrier frequency of $30 ~\rm{GHz}$, any user located at a distance of shorter than $100~\rm{m}$ is considered to lie in the near-field region of the aperture. 

When the user lies in the far-field region of the aperture, as explained in Sec. \ref{nearfield}, we partition the aperture into tiles such that the far-field condition is hold for each tile. In fact, by partitioning the aperture into tiles in the near-field, the signal wavefront can be seen as a plane wave for each tile. Although the elevation and azimuth angles of these impinging plane waves from the user to the tiles are different, the difference between the unknown parameters $(\alpha_{1_{n_x n_y}},\alpha_{2_{n_x n_y}})$ for neighboring tiles is small. Following the above example, assume that the distance between the user and the center of the aperture is $d_0=30~\rm{m}$ and the elevation and azimuth angles are $\theta = 30^{\circ}$ and $\phi=60^{\circ}$. Since the user lies in the near-field region of the aperture, we partition the aperture into four tiles of size $0.25~\rm{m} \times 0.25~\rm{m}$ where the far-field distance for each tile becomes $d_0=25~\rm{m}$ and far-field assumption holds. 
Then, the unknown parameters $(\alpha_1,\alpha_2)$, for these tiles are $(0.253,0.436)$, $(0.254,0.429)$, $(0.246,0.437)$, and $(0.247,0.430)$. We exploit the proximity between these unknown parameters in the proposed channel estimation scheme for near-field to reduce the estimation overhead.



For the near-field region, we use (\ref{activealpha_1t}) and (\ref{activealpha_2t}) to provide initial values for the iterative algorithm. Then, we apply Algorithm \ref{algorithm} to the tile close at the center of the aperture to estimate the unknown parameters of the tile $(\alpha_1,\alpha_2)$. For each of the remaining tiles, since the unknown parameters of each tile is very close to the parameters of its neighbouring tiles, we consider the average of estimated values of the neighbouring tiles as the initial values. Then, we apply Algorithm \ref{algorithm} to estimate the unknown parameters of each tile. Due to the existence of good initial values, the number of iteration in Algorithm \ref{algorithm} decreases, and hence the overhead of the proposed scheme will be reduced. 
After estimating the unknown parameters of the tiles, we send one more pilot for each tile to obtain the phase of the channel between each tile and the user.
Finally, we use (\ref{phase_shift-cont}) to configure the phase-shifting elements of all tiles, in which $\beta_{0_{n_xn_y}}$ is the phase of the channel between $(n_x,n_y)$-th tile and the user. 


%


\section{Numerical Results}
\label{num_result}
In this section, we provide the numerical results to evaluate the performance of the proposed channel estimation algorithm.

\subsection{Settings of the Numerical Experiments}
In this subsection, we explain the parameter setup of the numerical experiments. The simulation results have been averaged over 1000 random channel realization. In each channel realizations, while the distance between the user and the center of the BS is fixed, the elevation and azimuth angles of the LoS path follow the uniform distribution, i.e., $\theta \sim U(0,\pi/2), \phi \sim U(0,2\pi)$. In addition to the LoS path, we assume that there are 4 NLoS paths due to scatters between the user and the BS. The elevation and azimuth angles of each NLoS path from those scatters to the center of aperture follow the uniform distribution. Moreover, we model the path coefficient of each NLoS path as a complex Gaussian random variable, i.e., $\mathcal{CN}(0,\,\sigma^{2})$,  where $\sigma^2$ is  $20$ $\si{dB}$ weaker than the power of the LoS component \cite{wang2021joint}. 
Unless otherwise specified, the system parameters for numerical experiments are listed in Table \ref{table}.

\begin{table}
	\renewcommand{\arraystretch}{1.1}
	\caption{A list of system parameters for numerical experiments.}
	\label{table1}
	\centering
	\begin{tabular}{c||c||c}
		\hline
		\bfseries Parameters & \bfseries Values & \bfseries Description\\
	    \hline 
		$f_c$ & $30~\rm{GHz}$ & Carrier frequency\\
	    \hline 
		$\lambda$ & $1~\rm{cm}$ & Wavelength\\
	    \hline
		$M_x$ & $257$ & Number of antenna along $x$ axis\\
	    \hline
		$M_y$ & $ 257 $ & Number of antenna along $y$ axis\\
		\hline 
		$d_r$ & $\lambda/4$ & Unit element spacing\\
		\hline 
		$L_e$ & $0.8 d_r$ & Width and length of each phase-shifting element\\
		\hline
		$P$ & $30~\rm{dBm}$ & Transmission power of the BS during data transmission \\
		\hline
		$N_0$ & $-115~\rm{dBm}$ & Noise power for $200~\rm{KHz}$\\
		\hline
	\end{tabular}
	\label{table}
\end{table}

\subsection{Performance Evaluation}
In this section, we present the numerical results for the proposed algorithm and make
comparisons with other channel estimation schemes. Three benchmark schemes are considered for comparison, including a hierarchical search scheme, a CS-based channel estimation scheme, and localization-based channel estimation. 

\emph{Benchmark Scheme 1:}
 In the hierarchical search scheme, proposed in \cite{xiao2016hierarchical}, closed-form expressions are provided to generate a codebook consisting of codewords with different beam widths. In this scheme, joint sub-array and deactivation approach is exploited to design a binary-tree codebook. The pilot overhead of this hierarchical scheme is given by $2 \log_2{(M)}$, where $M$ is the number of phase-shifting elements at the aperture.
 
\emph{Benchmark Scheme 2:} 
In the CS-based channel estimation scheme, proposed in \cite{tsai2018efficient}, the problem is formulated as a sparse signal recovery problem. Then, the problem is solved by the orthogonal matching pursuit algorithm employing non-uniformly quantized angle grids. The pilot overhead of this scheme is given by $\mathcal{O}\left(L\ln{\left(M\right)}\right)$, where $L$ is the number of non-zero spatial channel paths.

\emph{Benchmark Scheme 3:} 
In the near-field of the BS, the channel is characterized by the location of the user. Therefore, we consider a localization scheme to estimate the channel by the location parameters. In \cite{guidi2021radio}, the location of a single antenna transmitter in the near-field is retrieved from the incident spherical wavefront.

\begin{figure}
\centering
    \includegraphics[scale=0.5]{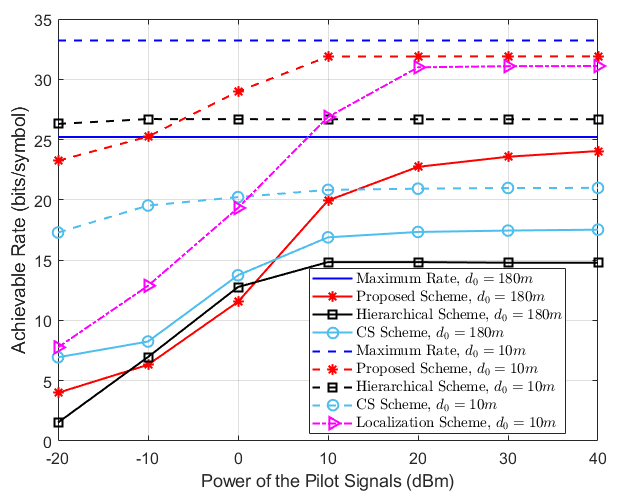}
\caption{\small{Achievable rate vs. the transmit power of the pilot signals of the proposed scheme and benchmark schemes for both far-field and near-field regions of the BS when the number of the pilot signals is fixed to 23.}}
\label{figure6}
\end{figure}

The BS uses the acquired CSI during the channel estimation period to maximize the received data rate by the user. Therefore, we consider the achieved data rate by the user using the acquired CSI as a performance metric.
The achieved data rate is calculated by
\begin{align}
  R = \log_2 \left( 1 + \frac{P_t\left|H\left(\bm{\beta}\right)\right|^2}{N_0}\right),
   \label{achieved rate}
\end{align}
where $N_0$ is the AWGN power, $P_t$ is the transmission power at the BS, and $\left|H\left(\bm{\beta}\right)\right|$ is the BS-user channel in (\ref{endtoend_exact_ITS}) which  depends on the phase configuration at the aperture. 
It is worthwhile to note that the imperfect CSI is used for the configuration of the elements of the aperture, whereas a high quality CSI of the scalar end-to-end channel (including beamforming at the transmitter) will be acquired at the user with almost perfect phase estimation to enable coherent communication.          

\begin{figure}
\centering
    \includegraphics[scale=0.5]{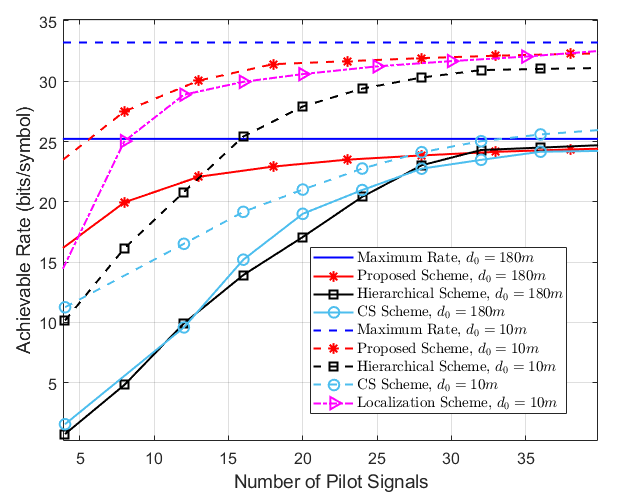}
\caption{\small{Achievable rate vs. the number of pilot signals of the proposed scheme and benchmark schemes for both far-filed and near-field regions of the BS when the transmit power of the pilot signals is fixed to $20~\rm{dBm}$.}}
\label{figure7}
\end{figure} 

Fig. \ref{figure6} illustrates the achieved data rate of the proposed scheme and the benchmark schemes as a function of the transmission power of the pilot signals. In addition, in this figure, we compare the performance of the proposed scheme with the benchmark schemes in both the far-field and near-field regions of the BS. 
In the far-field region of the BS, we set the maximum number of iteration of the proposed scheme to 4. For the near-field, we partitioned the aperture into 4 tiles. For each tile, we set the number of iteration to one.
Since each iteration requires five pilot signals and three pilot signals are required to provide the initial values, the maximum total number of pilot signals is 23. It can be
observed from Fig. \ref{figure6} that when the power of the pilot signals is low, the noise is comparable to the received pilot signals, and hence the proposed scheme cannot estimate the unknown parameters accurately. However, when the power of the pilot signal increases, the received pilot signals are much stronger than the noise, and hence the estimation error decreases. As illustrated in Fig. \ref{figure6}, the proposed
scheme, in general, achieves a significant gain over the other benchmark schemes since it exploits the specific structure of the radiated beam (the sinc function) in the LoS dominated channel model.

\begin{figure}
\centering
    \includegraphics[scale=0.5]{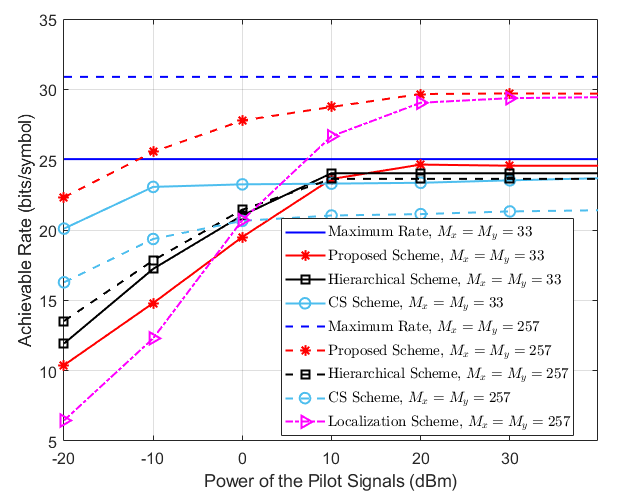}
\caption{\small{Achievable rate vs. the transmit power of the pilot signals of the proposed scheme and the benchmark schemes for different number of phase-shifting elements when the number of the pilot signals is fixed to 23.}}
\label{figure8}
\end{figure}
 
\begin{figure}
\centering
    \includegraphics[scale=0.5]{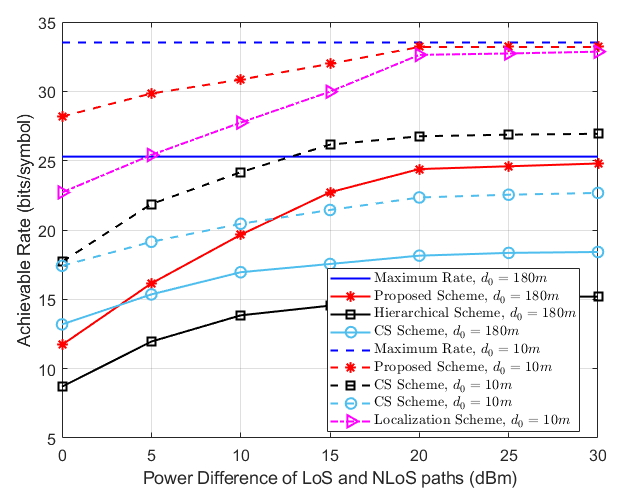}
\caption{\small{Achievable rate vs. power difference of LoS and NLoS paths of the proposed scheme and benchmark schemes for both far-field and near-field of the BS when the transmit power and the number of the pilot signals are fixed to $20~\rm{dBm}$ and 23, respectively.}}
\label{figure9}
\end{figure}

In Fig. \ref{figure7}, we show the achieved data rate of the proposed scheme and the benchmark schemes as a function of the number of pilot signals. In this figure, we fix the transmit power of the pilot signals to $20~\rm{dBm}$.  We observe from Fig. \ref{figure7} that the achieved data rate of the proposed and benchmark schemes increases with the number of pilot signals. In the far-field region of the BS, the proposed, hierarchical, and CS schemes can approximately achieve the maximum rate when the number of pilot signals is more than 30. In the near-field region of the BS, the achieved data rate of the CS scheme cannot increase more than a certain value due to the assumption of quantized values for $\alpha_1$ and $\alpha_2$. As illustrated in Fig. \ref{figure7}, the proposed scheme outperforms all benchmark schemes for the different number of pilot signals.

Fig. \ref{figure8} compares the achieved data rate by the proposed scheme with the benchmark schemes for two different numbers of phase-shifting elements at the aperture. Similar to Fig. \ref{figure6}, the number of pilot signals is fixed to 23.  Assuming perfect CSI, the achieved data rate has to increase with $M$. Nevertheless, the performance of the hierarchical scheme does not change with $M$. This is due to the fact that when $M$ increases, the signal beam width is narrower, and hence, more accurate estimations for $\alpha_1$ and $\alpha_2$ are required, which is not feasible with a low number of pilot signals. In addition, the performance of the CS scheme decreases with $M$ since the pilot overhead of the CS scheme is $\mathcal{O}\left(L\ln{\left(M\right)}\right)$. On the other hand, the achieved data rate of the proposed scheme increases with $M$ since accurate estimation for $\alpha_1$ and $\alpha_2$ can be obtained by the proposed scheme.

\begin{figure}
\centering
    \includegraphics[scale=0.5]{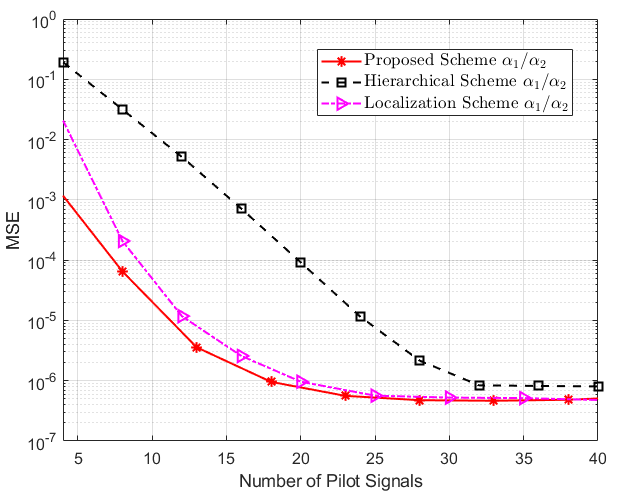}
\caption{\small{MSE vs. the number of pilot signals of the proposed scheme and the hierarchical scheme when the transmit power and the distance between the BS and the user are fixed to $20~\rm{dBm}$ and $25~\rm{m}$, respectively.}}
\label{figure10}
\end{figure}

\begin{figure}
\centering
    \includegraphics[scale=0.5]{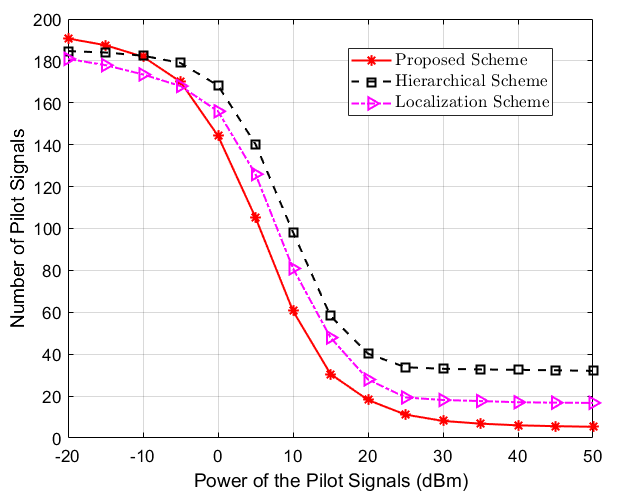}
\caption{\small{Number of pilot signals required to achieve the $\text{MSE}=10^{-6}$ vs. the transmit power of the pilot signals when the distance between the BS and the user are fixed to $25~\rm{m}$.}}
\label{figure11}
\end{figure}


Fig. \ref{figure9} illustrates the achieved data rate of the proposed scheme and the benchmark schemes as a function of the power difference of LoS and NLoS paths. The horizontal axis of this figure indicates that how much the power of NLoS path components is weaker than the power of the LoS component. In this figure, the transmit power and the number of the pilot signals are fixed to $20~\rm{dBm}$ and 23, respectively. In the proposed scheme, we consider the effect of NLoS path components as interference. As a result, when the power of NLoS path components decreases, the proposed scheme can estimate $\alpha_1$ and $\alpha_2$ more accurately. As shown in Fig. \ref{figure9}, the proposed scheme outperforms all benchmark schemes even when the power of NLoS path components is comparable to the power of LoS path component.

In Fig. \ref{figure10}, we compare the convergence behavior of the proposed scheme with the hierarchical and localization schemes. We consider the mean square error (MSE) to study the convergence behaviour, which is defined by $\text{MSE}_i= \mathbb{E}\{|\alpha_i-\hat{\alpha}_i|^2\}, \forall i\in\{1,2\}$. In this figure, we fix the transmit power and the distance between the BS and the user to $20~\rm{dBm}$ and $25~\rm{m}$, respectively. Since the CS scheme estimates the whole channel, not the AoA/AoD, we compare the proposed scheme only with the hierarchical and localization schemes. We observe from Fig. \ref{figure10} that when the number of pilot signals increases, the MSE of the proposed scheme and the benchmark schemes decreases.  In the hierarchical scheme, when the number of pilot signals increases, narrower beam widths are used to estimate $\alpha_1$ and $\alpha_2$ more accurately. However, since $M_x=M_y=257$, after $\lfloor \log_2{257} \rfloor= 8$ iteration, the MSE of the hierarchical scheme does not decrease. Fig. \ref{figure10} suggests that the proposed scheme and localization scheme can accurately estimate $\alpha_1$ and $\alpha_2$ when the number of pilot signals is more than 23 (4 iterations in Algorithm \ref{algorithm}).

In Fig. \ref{figure11}, we show the number of pilot signals required to achieve the desired MSE as a function of the transmission power of the pilot signals. In this figure, we fix the desired MSE and the distance between the BS and the user to $10^{-6}$ and $25~\rm{m}$, respectively. We observe from Fig. \ref{figure11} that the required number of pilot signals of the proposed and benchmark schemes decreases with the transmission power of the pilot signals. In addition, when the transmission power of the pilot signals is more than $-5~\rm{dBm}$, the proposed scheme needs less number of pilot signals compared to the benchmark scheme to achieve the $\text{MSE}=10^{-6}$. Fig. \ref{figure11} suggests that at a high signal to noise ratio (SNR) regime, the proposed scheme requires only five pilots to perfectly estimate the unknown channel parameters. This is because the proposed scheme exploits the specific structure of the radiated beams (two sinc functions) to estimate the unknown parameters.

\section{CONCLUSION}
In this paper, we proposed a channel estimation scheme for the holographic massive MIMO systems in the dominated LoS channel model. In the far-field region, we modeled the channel based on the path parameters of the system. We show that only two path parameters are required to be estimated to obtain the optimal phase shifts for all phase-shifting elements of the aperture. For the near-field, we first partitioned the aperture into tiles. Then, the channel was modeled as the superposition of the channels through the individual tiles. Moreover, only two parameters are required to be estimated for each tile. The proposed channel estimation scheme exploits the specific structure of the radiated beams (two sinc functions) to estimate the unknown parameters. 
The simulation results verified that the proposed scheme achieves significant performance gains over existing channel estimation schemes.

\section*{Appendix}
\subsection{Proof of Lemma \ref{lemma1}}
We start the proof by substituting the linear phase shift in (\ref{phase_shift_far}) into (\ref{endtoend_approx3}). We have
\begin{align}
   G^{(ff)}(\bm{\beta}) 
  \approx \left(\frac{ \sqrt{F}\lambda e^{-jk_0d_0}}{4\pi d_0}\right)   \sum_{m_x=-\frac{M_x-1}{2}}^{\frac{M_x-1}{2}}e^{jk_0d_r\Big(m_x\big(\alpha_1-\beta_1\big)\Big)} \!\!\sum_{m_y=-\frac{M_y-1}{2}}^{\frac{M_y-1}{2}} e^{jk_0d_r\Big(m_y\big(\alpha_2-\beta_2\big)\Big)}.
    \label{endtoend_approx3_app}
\end{align}
Using the sum of terms in a geometric progression, we can write 
\begin{align}
  \sum_{m=-\frac{M-1}{2}}^{\frac{M-1}{2}}e^{jma}=\frac{e^{-j\left(\frac{M-1}{2}\right)a}\left(1-e^{jMa}\right)}{1-e^{ja}}=
  \frac{\sin{\left(Ma/2\right)}}{\sin{\left(a/2\right)}}.
  \label{sum_geo}
\end{align}
 Now, using (\ref{sum_geo}) in (\ref{endtoend_approx3_app}), we have
 \begin{align}
    &G^{(ff)}(\bm{\beta})\approx \left(\frac{ \sqrt{F}\lambda e^{-jk_0d_0}}{4\pi d_0}\right)  \frac{\sin{\left(M_x\frac{k_0d_r}{2}\left(\alpha_1-\beta_1\right)\right)}}{\sin{\left(\frac{k_0d_r}{2}\left(\alpha_1-\beta_1\right)\right)}} \times \frac{\sin{\left(M_y\frac{k_0d_r}{2}\left(\alpha_2-\beta_2\right)\right)}}{\sin{\left(\frac{k_0d_r}{2}\left(\alpha_2-\beta_2\right)\right)}}.
    \label{ITS-user_1_app}
\end{align}
Substituting $M_x=L_x/d_r$ and $M_y=L_y/d_r$ into (\ref{ITS-user_1}) completes the proof.

	\bibliographystyle{IEEEtran}
	\end{document}